\documentclass[letter,11pt]{article}
\usepackage[letterpaper,text={6.6in,9.3in}]{geometry}

\usepackage{setspace}
\usepackage{authblk}
\usepackage{hyperref}
\usepackage{amsfonts}
\usepackage{amsmath}
\usepackage{amssymb}
\usepackage{amsthm}
\usepackage{graphicx}
\usepackage{color}
\usepackage[none]{hyphenat}
\usepackage{wrapfig}
\usepackage[justification=centering]{caption}
\usepackage{bm}
\usepackage{enumitem}
\usepackage{multirow}

\theoremstyle{plain}
\newtheorem{theorem}{Theorem}[section]

\newtheorem{corollary}[theorem]{Corollary}
\theoremstyle{definition}
\newtheorem{definition}[theorem]{Definition}

\usepackage{algorithm}
\usepackage[noend]{algpseudocode}

\theoremstyle{plain}
\theoremstyle{definition}

\graphicspath{ {/} }

\newcommand{\xmark}{$\times$\space}
\newcommand{\cmark}{\checkmark\space}

\setlist[enumerate]{itemsep=0mm,topsep=1mm}
\setlist[itemize]{itemsep=0mm,topsep=1mm}

\begin{document}
	
\pagenumbering{gobble}
\begin{titlepage}	
{ 
	\title{Cheating by Duplication: \\
		Equilibrium Requires Global Knowledge
		{\let\thefootnote\relax\footnote{
			This research was supported by the Israel Science Foundation (grant 1386/11).
		}}
		{\let\thefootnote\relax\footnote{\textbf{Contact:}
			\href{mailto:moshe.sulamy@cs.tau.ac.il}{moshe.sulamy@cs.tau.ac.il}}
		}
	}
	\author{Yehuda Afek}
	\author{Shaked Rafaeli}
	\author{Moshe Sulamy}
	\affil{Tel-Aviv University}
	\date{}
	\maketitle
}

	\begin{abstract}
		The question of what global information must distributed rational agents a-priori know about the network in order for equilibrium to be possible is researched here.
Until now, distributed algorithms with rational agents have assumed that $n$, the size of the network,
is a-priori known to the participants.
We investigate the above question, considering different distributed computing problems and showing how much each agent must a-priori know about $n$ in order for distributed algorithms to be equilibria. The main tool considered throughout the paper is the advantage an agent may gain by duplication- pretending to be more than one agent. 


We start by proving that when no bound on $n$ is given equilibrium for Coloring and Knowledge Sharing is impossible.
We provide new algorithms for both problems when $n$ \emph{is} a-priori known to all agents,
thus showing that there are algorithms in which the only way for an agent to gain an advantage is duplication.
We further show that for each distributed problem there is an a-priori known range, an upper and a lower bound on $n$, such that if the actual $n$ is guaranteed to lay in that range, equilibrium is possible.
By providing equilibria for a specific range, and impossibility results for any larger range, we prove the tight range necessary for equilibrium in: Leader Election, Knowledge Sharing, Coloring, Partition and Orientation.

	\end{abstract}

	\hfill
	\begin{center}
		{ \bfseries Regular submission }

		{ Eligible for best student paper award. }
	\end{center}
	
\end{titlepage}
\pagenumbering{arabic}

\section{Introduction}
\label{section_intro}

The complexity and simplicity of most distributed computing problems
depend on the inherent a-priori knowledge given to all participants.
Usually, the more information processors in a network start with,
the more efficient and simple the algorithm for a problem is.
Sometimes, this information renders an otherwise unsolvable problem, solvable. 

We consider a network of \emph{rational} agents \cite{Abraham:2011,DISC13/ADH}
who participate in an algorithm and may deviate from it if they deem a deviation more profitable for them,
i.e., the execution is more likely to output their desired output.
To differentiate from Byzantine faults,
we require the \emph{Solution Preference} property 
that ensures agents never prefer an outcome in which the algorithm fails
(e.g., terminates incorrectly) over an outcome with a legal output.
Previous  works in this setting \cite{DISC13/ADH,Afek:2014:DCB:2611462.2611481,Halpern:2016:RCE:2933057.2933088}
assumed agents a-priori know $n$, the number of agents in the network (henceforth called the \emph{actual} number).

Our model is motivated by multi-agent protocols in which the participants may cheat
in order to achieve the result \emph{they} think is best for them.
Consider a distributed frequency assignment (Coloring) between cellular network providers,
in which each prefers a certain frequency (color) for which it already has equipment or infrastructure.
Companies may then cheat in the distribution process in order for their preferred frequency to be assigned to them.
Another example is an online game,
in which players start by selecting the player that will host the game
and thus enjoy the best network latency in the game to follow.

In this paper we examine the a-priori knowledge about $n$ required for equilibrium in a distributed network of \emph{rational} agents, each of which has a preference over the output.
Unlike the case in which $n$ is known,
agents may also deviate from the algorithm by \emph{duplicating} themselves to affect the outcome.
This deviation is also known as a Sybil Attack \cite{Douceur:2002:SA:646334.687813},
commonly used to manipulate internet polls,
increase page rankings in Google \cite{Bianchini:2005:IP:1052934.1052938}
and affect reputation systems such as eBay \cite{Bhattacharjee:2005:ABS:1080192.1080203, Cheng:2005:SRM:1080192.1080202}.  In this paper, we use a Sybil Attack to prove impossibility of equilibria.
For each problem presented, an equilibrium when $n$ is known is provided or was provided in a previous work,
thus in these cases deviations that do not include duplication cannot benefit the agents.
Obviously, deviations from the algorithm that include both duplicating and additional cheating are also possible.

The problems we examine here can be solved in standard distributed computing without any knowledge about $n$,
since we can easily acquire the size of the network by a broadcast and echo.
However, learning the size of the network reliably is no longer possible with rational agents and thus,
for some problems, a-priori knowledge of $n$ is critical for equilibrium.

Intuitively, the more agents an agent is disguised as,
the more power to affect the output of the algorithm it has.
For every problem, we strive to find the maximum number of duplications a cheater may be allowed to duplicate without gaining the ability to affect the output, i.e., equilibrium is still possible.
This maximum number of duplications depends on whether other agents will detect that a duplication has taken place,
since the network could not possibly be this large.
To detect this situation they need to possess knowledge about the network size, or about a specific structure.

In this paper, we translate this intuition into a precise computation of the relation between the lower bound $\alpha$ and the upper bound $\beta = f(\alpha)$ on  $n$, that must be a-priori known in order for equilibrium to be possible. 
We denote this relation $f$-bound.
These bounds hold for both deterministic and non-deterministic algorithms.

To find the $f$-bound of a problem,
we first show what is the minimum number of duplications
for which equilibrium is impossible when $n$ is not known at all,
and then show an algorithm that is an equilibrium when the amount of duplications is
limited to that number, as well as when $n$ itself is a-priori known.
Finally, we calculate the $f$-bound by balancing out the profit an agent may gain by duplicating itself
against the risk it takes of being caught.

Table \ref{tbl:results} summarizes our contributions and related previous work (where there is a citation).
A \cmark mark denotes that we provide herein an algorithm for this case,
and an \xmark denotes that we prove that no equilibrium is possible in that case.
Known $n$ refers to algorithms in which $n$ is a-priori known to all agents.
Unknown $n$ refers to algorithm or impossibility of equilibrium when agents a-priori know \emph{no} bound on $n$.
The $f$-bound for each problem is a function $f$ for which there is an equilibrium when
the upper and lower bounds on $n$ satisfy $\alpha \leq \beta \leq f(\alpha)$,
and no equilibrium exists when $\beta > f(\alpha)$.
A problem is $\infty$-bound if there is an equilibrium given \emph{any} finite bound,
but no equilibrium exists if no bound or information about $n$ is a-priori given.
A problem is \emph{unbounded} if there is an equilibrium even when neither $n$ nor any bound on $n$ is given.

\begin{table}[ht]
	\centering
	\begin{tabular}{|c|c|c|c|}
		\hline
		Problem & Known $n$ & Unknown $n$ & $f$-bound  \\ \hline
		Coloring
		& \checkmark
		& \xmark
		& $\infty$*
		\\ \hline
		\shortstack{Leader Election\\\space}
		& \shortstack{\checkmark \\ ADH'13 \cite{DISC13/ADH}}
		& \shortstack{\xmark \\ ADH'13 \cite{DISC13/ADH}}
		& \shortstack{$(\alpha+1)$ \\\space}
		\\ \hline
		Knowledge Sharing
		& \multirow{2}{*}{\shortstack{\checkmark\\ AGLS'14 \cite{Afek:2014:DCB:2611462.2611481}}}
		& \multirow{2}{*}{\xmark}
		& $(2\alpha-2)$*
		\\ \cline{1-1} \cline{4-4}
		$2$-Knowledge Sharing
		&
		&
		& $\infty$*
		\\ \hline
		Partition, Orientation
		& \checkmark
		& \checkmark
		& Unbounded
		\\ \hline
	\end{tabular}
	\caption{Summary of paper contributions,
		equilibria and impossibility results for different problems with different a-priori knowledge about $n$
		\\\textbf{*} $f$-bound proven for a ring graph}
	\label{tbl:results}
\end{table}

\subsection{Related Work}

The connection between distributed computing and game theory
stemmed from the problem of secret sharing \cite{CACM/Shamir79}.
Further works continued the research on secret sharing and multiparty computation
when both Byzantine and rational agents are present
\cite{PODC/AbrahamDGH06, PODC/DaniMRS11, TCC/FuchsbauerKN10, SCN/GordonK06,  ICALP/GroceKTZ12, CYPTO/LysyanskayaT06}. 

Another line of research presented the BAR model (Byzantine, acquiescent \cite{OPODIS/WongLACD11} and rational) \cite{SOSP/AiyerACDMP05, PODC/MoscibrodaSW06, OPODIS/WongLACD11},
while a related line of research discusses converting solutions with a mediator to
cheap talk \cite{PODC/AbrahamDGH06, TCC/AbrahamDH08, Barany1992, JET/Ben-Porath03,  CRYPTO/DodisHR00,
	PODC/LepinskiMP04, TARK/McGrewPS03, TCS/Shoham05, ECON/UAV02, ET/UA04}.

Abraham, Dolev, and Halpern \cite{DISC13/ADH} were the first to present
protocols where processors in the network behave as rational agents,
specifically protocols for Leader Election.
In \cite{Afek:2014:DCB:2611462.2611481} the authors continue this line
of research by providing basic building blocks for game theoretic distributed algorithms,
namely a wake-up and knowledge sharing equilibrium building blocks.
Algorithms for consensus, renaming, and leader election are presented using these building blocks.
Consensus was researched further by Halpern and Vilacça \cite{Halpern:2016:RCE:2933057.2933088},
who showed that there is no ex-post Nash equilibrium, and a Nash equilibrium that tolerates $f$ failures
under some minimal assumptions on the failure pattern.

Coloring and Knowledge Sharing have been studied extensively in a distributed setting
\cite{AttiyaBook,Awerbuch:1989:NDL:1398514.1398717,Cole:1986:DCT:10366.10368,
	Kuhn:2006:CDG:1146381.1146387,Linial1986,Linial:1987:DGA:1382440.1382990,
	Szegedy:1993:LBG:167088.167156}.
An algorithm for Knowledge Sharing  with rational agents was presented in \cite{Afek:2014:DCB:2611462.2611481}, while
Coloring with rational agents has not been studied previously, to the best of our knowledge.

\section{Model}
\label{section_model}
We use the standard message-passing model, where the network is a bidirectional graph
$G=\left(V,E\right)$ with $n$ nodes, each node representing an agent,
and $|E|$ edges over which they communicate. 
$G$ is assumed to be $2$-vertex-connected\footnotemark.
Throughout the entire paper, $n$ always denotes the actual number of nodes in the network.
\footnotetext{
	This property was shown necessary in \cite{Afek:2014:DCB:2611462.2611481}, since if such a node exists it can alter any message passing through it.
	Such a deviation cannot be detected since all messages between the sub-graphs this node connects
	must traverse through it.
	This node can then skew the algorithm according to its preferences.
}

Initially, each agent knows its \emph{id} and input,
but not the \emph{id} or input of any other agent. 
We assume the prior of each agent over any information it does not know is uniformly distributed over all possible values. 
Each agent is assigned a unique \emph{id},
taken from the set of natural numbers.
Furthermore, we assume all agents start the protocol together, i.e., all agents wake-up at the same time.
If not, we can use the Wake-Up \cite{Afek:2014:DCB:2611462.2611481} building block to relax this assumption. 

\subsection{Equilibrium in Distributed Algorithms}
\label{section_legal_correct}

Informally, a distributed algorithm is an equilibrium if no agent at no point in the execution can do better by unilaterally deviating from the algorithm. 
When considering a deviation, an agent assumes all other agents follow the algorithm,
i.e., it is the only agent deviating.

Following the model in \cite{DISC13/ADH, Afek:2014:DCB:2611462.2611481}
we assume an agent always aborts the algorithm whenever it detects a deviation made by another agent,
even if the detecting agent can gain by not aborting.  
In this manner, our notion of equilibrium is basically Bayes-Nash equilibrium\footnotemark,
but not sequential equilibrium \cite{KrepsWilson82}.
We elaborate on this difference in the Discussion (Section~\ref{section_discussion}).

\footnotetext{
	In \cite{DISC13/ADH}, sequential equilibrium is shown by an additional assumption on the utility function;
	however, if an agent detects a deviation that does not necessarily lead to the algorithm failure,
	it is still not a sequential equilibrium.
}

Each node in the network is a \emph{rational} agent,
following the model in \cite{Afek:2014:DCB:2611462.2611481}.
The algorithms produce a single output per agent, once, at the end of the execution.
Each agent has a preference only over its own output.

Formally, let $o_a$ be the output of agent $a$, let $\Theta$ be the set of all possible output vectors,
and denote the output vector $O =\{o_1,\dots,o_n\} \in \Theta$, where $O[a]=o_a$.
Let $\Theta_L$ be the set of \emph{legal} output vectors, in which the protocol terminates successfully,
and let $\Theta_E$ be the set of \emph{erroneous} output vectors, such that
$\Theta = \Theta_L \cup \Theta_E$ and $\Theta_L \cap \Theta_E = \varnothing$.

Each agent $a$ has a utility function $u_a: \Theta \rightarrow \mathbb{N}$.
The higher the value assigned by $u_a$ to an output vector,
the better this vector is for $a$.
To differentiate rational agents from Byzantine faults,
we assume the utility function satisfies \emph{Solution Preference} \cite{DISC13/ADH,Afek:2014:DCB:2611462.2611481}
which guarantees that an agent never has an incentive to cause the algorithm to fail.
                                                                                     
\begin{definition}[Solution Preference]                                              
	The utility function $u_a$ of an agent $a$ never assigns a higher utility           
	to an erroneous output than to a legal one, i.e.:                                   
	$$\forall{a,O_L \in \Theta_L,O_E \in \Theta_E}: u_a(O_L) \geq u_a(O_E)$$
\end{definition}

We differentiate the \emph{legal} output vectors, which ensure the output is valid and not erroneous, from the \emph{correct} output vectors, which are output vectors that are a result of a correct execution of the algorithm, i.e., without any deviation.
The Solution Preference guarantees agents never prefer an erroneous output. However, they may prefer a \emph{legal} but \emph{incorrect} output.

Recall that we assume agents only have preferences over their own output,
i.e., for any $O_1,O_2 \in \Theta_L$ where $O_1[a]=O_2[a]$, $u_a(O_1) = u_a(O_2)$.
For simplicity, we also assume each agent $a$ has a \emph{single} preferred output value $p_a$,
and we normalize the utility function values, such that\footnotemark:
\begin{equation}
\label{eq_u_a}
u_a(O)=\begin{cases}
	1 & o_a=p_a$ and $O \in \Theta_L
	\\ 0 & o_a \neq p_a$ or $O \in \Theta_E\end{cases}
\end{equation}

Our results hold for \emph{any} utility function that satisfies Solution Preference.

\footnotetext{
	This is the weakest assumption that satisfies Solution Preference, 
	since it gives cheating agents the highest incentive to deviate.
	A utility assigning a lower value for failure than $o_a \neq p_a$
	would deter a cheating agent from deviating.
}

\begin{definition} [Expected Utility]
Let $r$ be a round in a specific execution of an algorithm. Let $a$ be an arbitrary agent. For each possible output vector $O$,
let $x_O(s,r)$ be the probability, estimated by agent $a$ at round $r$,
that $O$ is output by the algorithm if $a$ takes step $s$~\footnotemark,
and all other agents follow the algorithm.
The \emph{Expected Utility} $a$ estimates for step
$s$ in round $r$ of that specific execution is:
$$ \mathbb{E}_{s,r} [u_a]=\sum\limits_{O \in \Theta} x_O(s,r) u_a(O) $$
\end{definition}

\footnotetext{
	A step specifies the entire operation of the agent in a round.
	This may include drawing a random number, performing any internal computation,
	and the contents and timing of any message delivery.
}

Note that agents can also estimate the expected utility of \emph{other} agents by simply considering
a different utility function.

An agent will deviate whenever the deviating step leads to a strictly
higher expected utility than the expected utility of the next step of the algorithm.
By the utility function \ref{eq_u_a}, an agent will prefer \emph{any} deviating step that increases the probability of getting its preferred output, even if that deviating step also increases the risk of an erroneous output.

Let $\Lambda$ be an algorithm.
If by deviating from $\Lambda$ and taking step $s$,
the expected utility of $a$ is higher, we say that agent $a$ has an \emph{incentive to deviate} (i.e., cheat).
For example, at round $r$ algorithm $\Lambda$ may dictate that $a$ flips a fair binary coin
and sends the result to all of its neighbors.
Any other action by $a$ is considered a \emph{deviation}:
whether the message was not sent to all neighbors,
sent later than it should have, or whether the coin toss was not fair,
e.g., $a$ only sends $0$ instead of a random value. 
If no agent can unilaterally increase its expected utility by deviating from $\Lambda$,
we say that the protocol is an \emph{equilibrium}.
We assume a \emph{single} deviating agent, i.e., there are no coalitions of agents.

\begin{definition}[Distributed Equilibrium]
	Let $s(r)$ denote the next step of algorithm $\Lambda$ in round $r$.
	$\Lambda$ is an equilibrium if for any deviating step $s$,
	at any round $r$ of every possible execution of $\Lambda$:
	$$\forall{a,r,s}: \mathbb{E}_{s(r),r}[u_a] \geq \mathbb{E}_{s,r}[u_a]$$
\end{definition}

\subsection{Knowledge Sharing}
The Knowledge Sharing problem (adapted from \cite{Afek:2014:DCB:2611462.2611481}) is defined as follows:
\begin{enumerate}
	
	\item Each agent $a$ has a private input $i_a$, in addition to its $id$,
	and a function $q$, where $q$ is identical at all agents.
	
	\item A Knowledge Sharing protocol terminates \emph{legally} if all agents
	output the \emph{same} value, i.e., $\forall{a,b}: o_a=o_b \neq \bot$. 
	Thus the set $\Theta_L$ is defined as:
	$ O \in \Theta_L \iff \forall{a,b}: O(a) = O(b) \neq \bot$.
	
	\item A Knowledge Sharing protocol terminates \emph{correctly} (as described in Section~\ref{section_legal_correct}) if
	each agent outputs at the end the value $q(I)$
	over the input values $I=\{i_1,\dots,i_n\}$ of all other agents\footnotemark.
	\footnotetext{Notice that any output is legal as long as it is the output of all agents,
	but only a single output value is considered \emph{correct} for a given input vector.}

	\item The function $q$ satisfies the Full Knowledge property:
	\begin{definition}[Full Knowledge Property]
		\label{full_knowledge}
		A function $q$ fulfills the \emph{full knowledge} property if, 
		for each agent that does not know the input values of all other agents,
		any output in the range of $q$ is \emph{equally} possible.
		Formally, for any $1 \leq j \leq m$, fix $(x_1,\dots,x_{j-1},x_{j+1},\dots,x_m)$
		and denote $z_y = |\{x_j | q(x_1,\dots,x_j,\dots,x_m)=y\}|$.
		A function $q$ fulfills the \emph{full knowledge} property if, for any possible output $y$ in the range of $q$,
		$z_y$ is the same\footnotemark.
		
		\footnotetext{The definition assumes input values are drawn uniformly, otherwise the definition
			of $z_y$ can be expanded to the sum of probabilities over every input value for $x_j$.}
	\end{definition}
\end{enumerate}

We assume that each agent $a$ prefers a certain output value $p_a$.

\paragraph{$2$-Knowledge Sharing}
	The $2$-Knowledge Sharing problem is a Knowledge Sharing problem with exactly
	$2$ distinct possible output values.

\subsection{Coloring}
In the Coloring problem \cite{Cole:1986:DCT:10366.10368,Linial1986},
the output $o_a$ of each agent $a$ is a color.
$\Theta_L$ is any $O$ such that
$\forall{a} : o_a \neq \bot$ and  $\forall{(a,b) \in E} : o_a \neq o_b$.

We assume that every agent $a$ prefers a specific color $p_a$.

\section{Impossibility With No Knowledge}
\label{section_imp}

Here we show that without any a-priori knowledge about $n$,
there is no algorithm that is an equilibrium for both Knowledge Sharing and Coloring.

A fundamental building block in many algorithms is~\emph{Wake-Up} \cite{Afek:2014:DCB:2611462.2611481},
in which agents learn the graph topology. 
If $n$ is not known at all, how can an agent be sure the topology it learned is correct?

Let $a$ be a malicious agent with $\delta$ outgoing edges.
A possible deviation for $a$ is to simulate imaginary agents $a_1$, $a_2$
and to answer over some of its edges as $a_1$, and over the others as $a_2$,
as illustrated in Figure~\ref{figure:duplicated_node}.
From this point on $a$ acts as if it is 2 agents. 
Here we assume that the \emph{id}~space is much larger than $n$, 
allowing us to disregard the probability that the fake \emph{id} collides with an existing \emph{id}.

Note that an agent may be forced by the protocol to commit to its fake duplication early,
by starting the algorithm with a process that tries to map the graph topology,
such as the Wake-Up \cite{Afek:2014:DCB:2611462.2611481} algorithm.
If an algorithm does not begin by mapping the topology, an agent could begin the protocol as a single agent and duplicate itself at a later stage.
This allows the agent to "collect information" and increase its ability to effect the output.
Since a Wake-Up protocol can be added at the beginning of every algorithm,
we assume every algorithm starts by mapping the network,
thus forcing a duplicating agent to commit to its duplication scheme at the beginning of the algorithm.

\begin{figure}[H]
	\caption{Agent $a$ acting as separate agents $a_1$, $a_2$}	
	\label {figure:duplicated_node}
	\captionsetup{justification=centering}
	\includegraphics[scale=0.4]{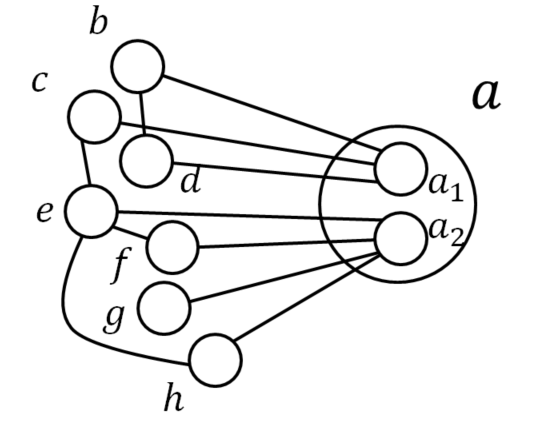}
	\centering
\end{figure}

Regarding the output vector, notice that an agent that pretends to be more than one agent
still outputs a \emph{single} output at the end.
The duplication causes agents to execute the algorithm as if it is executed on a graph $G'$ (with the duplicated agents) instead of the original graph $G$;
however, the output is considered legal if $O = \{o_a,o_b,\dots\} \in \Theta_L$
rather than if $\{o_{a_1},o_{a_2},o_b,\dots\} \in \Theta_L$.

It is important to emphasize that for any non-trivial distributed algorithm,
the outcome cannot be calculated using only private data without communication, i.e.,
for rational agents, no agent can calculate the outcome privately at the beginning of the algorithm.
This means that at round $0$,
for any agent $a$ and any step $s$ of the agent that does not necessarily result in algorithm failure,
it must hold that: $\mathbb{E}_{s,0}[u_v] \notin \{0,1\}$
(a value of $0$ means an agent will surely not get its preference,
and $1$ means it is guaranteed to get its preference).

In this section we label agents in graph $G$ as $a_1,...,a_n$, set in a clockwise manner in a ring and
arbitrarily in any other topology. These labels are not known to the agents themselves.


\subsection{Impossibility of Knowledge Sharing}
\label{know_share}


\begin{theorem}
	\label{theorem:KSgeneral}
	There is no algorithm for Knowledge Sharing
	that is an equilibrium in a $2$-connected graph when agents have no a-priori knowledge of $n$.
\end{theorem}

\begin{proof}
	Assume by contradiction that  $\Lambda$ is a Knowledge Sharing algorithm
	that is an equilibrium in any graph without knowing $n$. 
	Let $D$, $E$ be two $2$-connected graphs of rational agents.
	Consider the execution of $\Lambda$ on graph $H$ created by $D,E$, and adding two nodes $a_1,a_2$ and connecting
	these nodes to $1$ or more arbitrary nodes in both $D$ and $E$ (see Figure~\ref{figure:general_graph_clear}).

	\begin{figure}[H]
		\centering
		\begin{minipage}{0.45\textwidth}
			\centering
			\includegraphics[width=0.9\textwidth]{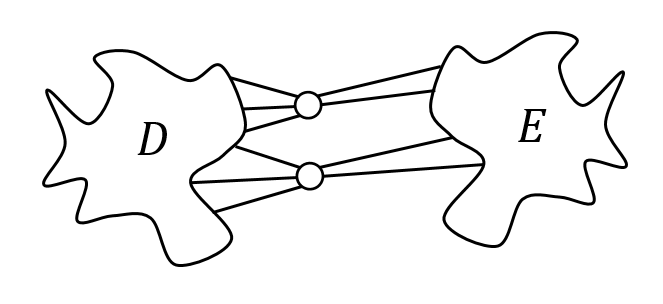} 
			\caption{Graph $H$ created by two arbitrary sub-graphs $D$,$E$}
			\label{figure:general_graph_clear}
		\end{minipage}\hfill
		\begin{minipage}{0.45\textwidth}
			\centering
			\includegraphics[width=0.9\textwidth]{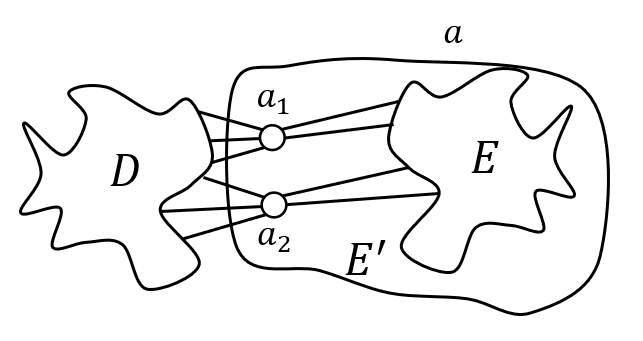}
			\caption{Example of agent $a$ pretending to be $E'=E \cup \{a_1,a_2\}$}
			\label{figure:general_graph_duplicated}
		\end{minipage}
	\end{figure}
	
	 
	Recall that the vector of agents' inputs is denoted by $I = i_1 , i _2 , \cdots , i_n$,
	and $n=|H|=|D|+|E|+2$.
	Let $t_D$ be the first round after which $q(I)$ can be calculated
	from the collective information that all agents in $D$ have\footnotemark,
	and similarly $t_E$ the first round after which $q(I)$ can be calculated in $E$.
	Consider the following three cases:
	
	\footnotetext{
		Regardless of the complexity of the computation.
	}
	
	\begin{enumerate}
		\item $\bm{t_E<t_D}$: $q(I)$ cannot yet be calculated in $D$ at round $t_E$.
		Let $E'=E \cup \{a_1,a_2\}$.
		Since $E \subset E'$,
		the collective information in $E'$ at round $t_E$ is enough to calculate $q(I)$.
		Since $n$ is not known, an agent $a$ could emulate the behavior of $E'$,
		making the agents believe the algorithm runs on $H$ rather than $D$.
		In this case, this cheating agent knows at round $t_E$ the value of $q(I)$ in this execution,
		but the collective information of agents in $D$ is not enough to calculate $q(I)$,
		which means the output of agents in $D$ still depends on messages from $E'$, the cheater.
		Thus, if $a$ learns that the output $q(I) \neq p_a$,
		it can send messages that may cause the agents in $D$ to decide a value $x \neq q(I)$.
		In the case where $p_a=x$, agent $a$ increases its expected utility by sending
		a set of messages different than that decreed by the protocol.
		Thus, agent $a$ has an incentive to deviate, contradicting distributed equilibrium.
		\label{case:t_great_full_graph}
		
		\item $\bm{t_D=t_E}$: both $E$ and $D$ have enough collective information to calculate $q(I)$
		at the same round.
		The collective information in $E$ at round $t_E$ already exists in $E'$ at round $t_E-1$.
		Since $t_D=t_E$, the collective information in $D$ is not enough to calculate $q(I)$ in round $t_E-1$.
		Thus, similarly to Case~\ref{case:t_great_full_graph}, $a$ can emulate $E'$ and has an incentive to deviate.
		
		\item $\bm{t_E>t_D}$: Symmetric to Case~\ref{case:t_great_full_graph}.
	\end{enumerate}
	
	Thus, $\Lambda$ is not an equilibrium for the Knowledge Sharing problem.
\end{proof} 

When $|D|=|E|$, the proof of Theorem~\ref{theorem:KSgeneral} brings us to the following corollary:
\begin{corollary}
	\label{cor:ks-limit}
	When a cheating agent pretends to be more than $n$ agents,
	there is no algorithm for Knowledge Sharing that is an equilibrium
	when agents have no a-priori knowledge of $n$.
\end{corollary}


\subsection{Impossibility of Coloring in a Ring}

The proof of Theorem~\ref{theorem:KSgeneral} relies on the Full Knowledge property of the Knowledge Sharing
problem, i.e., no agent can calculate the output before knowing all the inputs.
Recall that the Coloring problem, however, is a more local problem \cite{Nati:Locality}, and
nodes may color themselves without knowing anything about distant nodes.

\begin{theorem}
	\label{th_color_imp}
	There is no algorithm for Coloring that is an equilibrium in a $2$-connected graph
	when agents have no a-priori knowledge of $n$.
\end{theorem}

\begin{proof}
	In order to show an incentive to deviate, we generalize the notion of \emph{expected utility}.
	Recall that an agent outputs a \emph{single} color, even if it pretends to be several agents.
	In Coloring, a cheating agent only wishes to influence the output color of its \emph{original neighbors}
	to enable it to output its preferred color while maintaining the legality of the output.

	\begin{definition} [Group Expected Utility]
		Let $r$ be a round in an execution $\varepsilon$, and let $M$ be a group of agents.
		For any set $S=\{s_1,\dots,s_{|M|}\}$ of steps of agents in $M$,
		let $\Psi$ be the set of all possible executions for which the same messages traverse the links
		that income and outgo to/from $M$ as in $\varepsilon$ until round $r$,
		and in round $r$ each agent in $M$ takes the corresponding step in $S$.
		For each possible output vector $O$, let $x_O$ be the sum of probabilities over $\Psi$
		that $O$ is decided by the protocol.
		For any agent $v$, the \emph{Group Expected Utility} of $u_v$ by $M$ taking steps $S$ at round $r$
		in execution $\varepsilon$ is:
		$ \mathbb{E}_{M,S,r} [u_v]=\sum\limits_{O \in \Theta} x_O u_v(O) $.
	\end{definition}
	
	Assume by contradiction that $\Gamma$ is a Coloring algorithm that
	is an equilibrium in a ring with $n$ agents $\{a_1,\dots,a_n\}$.
	Let $G$ be a ring with a segment of $k$ consecutive agents, $k \geq 3$,
	all of which have the same color preference $p$. 
	Assume w.l.o.g., they are centered around $a_n$ if $k$ is odd and around $a_n,a_1$ if even.
	Let $L$ be the group of agents $\{ a_{n-1},\dots, a_{\lfloor \frac{n}{2}\rfloor + 1}\}$,
	and $R$ the group of agents $\{ a_1, \dots , a_{\lceil \frac{n}{2} \rceil - 1}\}$. 
	Denote $L' = L \cup \{a_{\lceil \frac{n}{2} \rceil}, a_n\}$ and $R' = V\setminus L'$
	(see Figures~\ref{figure:coloring_ring}, and~\ref{figure:coloring_ring_extended}).
		
	\begin{figure}[H]
		\centering
		\begin{minipage}{0.45\textwidth}
			\centering
			\includegraphics[width=0.9\textwidth]{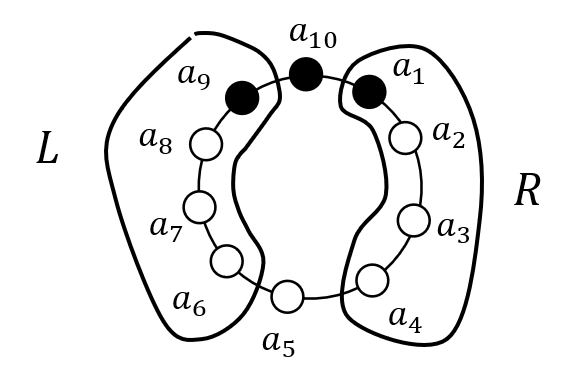} 
			\caption{Ring with 3 colliding agents $a_9, a_{10}, a_1$ with groups $L$, $R$}
			\label{figure:coloring_ring}
		\end{minipage}\hfill
		\begin{minipage}{0.45\textwidth}
			\centering
			\includegraphics[width=0.9\textwidth]{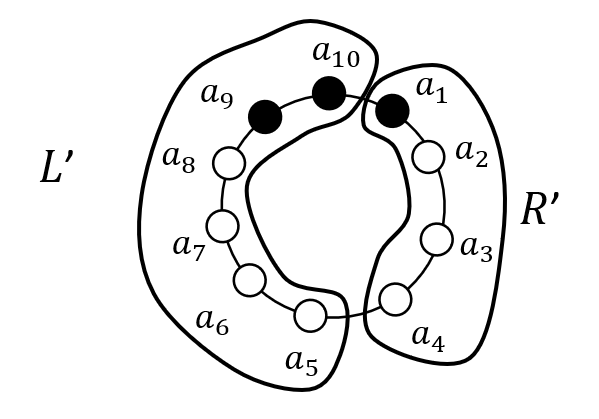}
			\caption{Ring with 3 colliding agents with groups $L'$, $R'$}
			\label{figure:coloring_ring_extended}
		\end{minipage}
	\end{figure}

	\begin{definition}
		Let $Y$ be a group of agents (e.g., $L$ or $R$). In any round $r$ in an execution,
		let $S^r(Y)$ denote the vector of steps of agents in $Y$ according to the protocol.
		We say $Y$ \emph{knows} the utility of agent $a$ if it holds that
		$\mathbb{E}_{Y,S^r(Y)}[u_{a}] \in \{0,1\}$.
		We say $Y$ \emph{does not know} the utility of agent $a$ if $0<\mathbb{E}_{Y,S^r(Y)}[u_{a}]<1$.	
	\end{definition}
	
	Recall that at round $0$ no agent (or group of agents) knows its utility or the utility of any other agent.
	Consider an execution of $\Gamma$ on ring $G$ and the groups $L,R$ in the following cases:

	\begin{enumerate}
	
	\item $L$ does not know $u_{a_n}$ throughout the entire execution of the algorithm,
	i.e., for agents in $L$ it holds that $0 < Pr[o_n \neq p] < 1$.
	Then if $L$ is emulated by a cheating agent,
	it has an incentive to deviate and set its output to $p$ (as otherwise its utility is guaranteed to be $0$).
	

	\item $L$ knows $u_{a_n}$ at some round $t_L$,
	and $R$ does not know $u_{a_n}$ before round $t_L$.
	Consider round $t_L-1$ and group $L'$:
	In round $t_L$, $L$ knows the utility of $a_n$, 
	thus the collective information of agents in $L$ at round $t_L$ already exists in $L'$ at round $t_L-1$.
	If $L'$ knows that $u_{a_n}=1$, then it had already won; otherwise, $L'$ knows that $u_{a_n}=0$.
	Consider the group $R' \subset R$, that does not know $u_{a_n}$ at round $t_L-1$. 
	If $L'$ is emulated by a cheating agent $a$,
	it can send messages that increase its probability to output $p$, thus increasing its expected utility.
	\label{coloring:case_L_R}
	
	\item $R$ knows $u_{a_n}$ before round $t_L$: symmetric to Case~\ref{coloring:case_L_R}.

	\end{enumerate}

	By the contradictory example for a ring, there is no equilibrium for Coloring for $2$-connected graphs,
	thus $\Gamma$ is not an equilibrium for the Coloring problem.
\end{proof}

\section{Algorithms}
\label{section_alg}

Here we present algorithms for Knowledge Sharing (Section~\ref{section:alg:ks})
and Coloring (Section~\ref{section:alg:coloring}).
The Knowledge Sharing algorithm is an equilibrium in a ring
when no cheating agent pretends to be more than $n$ agents.
The Coloring algorithm is an equilibrium in any $2$-connected graph
when agents a-priori know $n$.

Using an algorithm as a subroutine is not trivial in this setting,
even if the algorithm is an equilibrium,
as the new context as a subroutine may allow agents to deviate
towards a different objective than was originally proven.
Thus, whenever a subroutine is used, its equilibrium should be justified.

The full descriptions and proofs of the algorithms and Theorem~\ref{theorem:ks-limit}
can be found in Appendix~\ref{appendix_algs},
in addition to a Coloring algorithm with improved time complexity.

\subsection{Knowledge Sharing in a Ring}
\label{section:alg:ks}

First we describe the \texttt{Secret-Transmit} building block in which agent $a$ delivers its input $i_a$
at round $r$ to some agent $b$,
and no other agent in the ring learns any information about this input.

Agent $a$ selects a random number $R$ and let $X=R \oplus i_a$.
It then sends $R$ clockwise and $X$ counter-clockwise until each reaches the agent before $b$.
At round $r-1$, each neighbor of $b$ simultaneously sends $b$ the value it received, either $X$ or $R$.

We assume a global orientation around the ring.
This assumption can be easily relaxed via Leader Election \cite{Afek:2014:DCB:2611462.2611481},
which is an equilibrium in this application since the orientation
has no effect on the output.
The algorithm works as follows:

All agents execute Wake-Up \cite{Afek:2014:DCB:2611462.2611481} to learn $n'$,
the size of the ring which may include duplications.
For each agent $a$, denote $b^1_a$ the clockwise neighbor of $a$,
and $b^2_a$ the agent at distance $\lfloor \frac{n'}{2} \rfloor$ counter-clockwise from $a$.
All agents around the ring simultaneously use \texttt{Secret-Transmit},
each to transmit its input secretly to its corresponding $b^1$ and $b^2$
so all arrive at their destination at the same round $r=n'$.
At round $n'+1$, each agent sends its input around the ring.

\begin{theorem}
	\label{theorem:ks-limit}
	In a ring, the algorithm above is an equilibrium
	when no cheating agent pretends to be more than $n$ agents.
\end{theorem}

\subsection{Coloring}
\label{section:alg:coloring}

\begin{enumerate}
	\item All agents execute Renaming \cite{Afek:2014:DCB:2611462.2611481}
	which gives new names $1,\dots,n$ to the agents.
	Since agents strive to minimize their name's numeral value, Renaming is still an equilibrium.
	\item Each agent, in order of the new names, picks its preferred color if available,
	or the minimal available color otherwise, and sends its color to all of its neighbors.
\end{enumerate}

\section{How Much Knowledge Is Necessary?}
\label{section_approx}


Here we examine the effects of a-priori knowledge that \emph{bound} the possible value of $n$. 
We show that the possibility of algorithms that are equilibria depends on the range $[\alpha,\beta]$ in which $n$ might be, and show these ranges for different problems.

Table~\ref{tbl:summary} summarizes our results.
Partition and Orientation have equilibria without any knowledge of $n$;
however, the former is constrained to even-sized rings,
and the latter is a trivial problem in distributed computing
(radius $1$ in the \textit{LOCAL} model \cite{Linial:1987:DGA:1382440.1382990}).

\begin{definition}[$(\alpha,\beta)$-Knowledge]
	\label{def_alphabetaknowledge}
	We say agents have \emph{$(\alpha,\beta)$-Knowledge} about the actual number of agents $n$,
	$\alpha \leq \beta$, if all agents know that the value of $n$ is in $[\alpha, \beta]$.
	We assume agents have no information about the distribution over $[\alpha,\beta]$,
	i.e., they assume it is uniform.
\end{definition}

\begin{definition} [$f$-Bound]
	Let $f:\mathbb{N} \rightarrow \mathbb{N}$. A distributed computing problem $\mathbb{P}$ is \emph{$f$-bound} if: 
	\begin{itemize}
		\item There exists an algorithm for $\mathbb{P}$ that is an equilibrium 
		given $(\alpha,\beta)$-Knowledge for any $\alpha,\beta$ such that $ \beta \leq f(\alpha)$. 
		
		
		\item For any algorithm for $\mathbb{P}$, there exist $\alpha,\beta$ where $\beta > f(\alpha)$
		such that given $(\alpha,\beta)$-Knowledge the algorithm is not an equilibrium.
	\end{itemize}
\end{definition}

In other words, a problem is $f$-bound if given $(\alpha,\beta)$-Knowledge,
there is an equilibrium when $\beta \leq f(\alpha)$,
and there is no equilibrium when $\beta > f(\alpha)$. 
A problem is $\infty$-bound if there is an equilibrium given \emph{any} bound $f$,
but there is no equilibrium with $(1,\infty)$-Knowledge.
A problem is \emph{unbounded} if there is an equilibrium with $(1,\infty)$-Knowledge.

\begin{table}[h]
	\centering
	\begin{tabular}{|c|c|}
		\hline
		Bound & Problem (in a ring)  \\ \hline
		$\alpha+1$ & Leader Election\footnotemark \\ \hline
		$2\alpha-2$ & Knowledge Sharing \\ \hline
		$\infty$ & Coloring, $2$-Knowledge Sharing \\ \hline
		$unbounded$ & Partition, Orientation\textsuperscript{\ref{fn:general_graph}} \\ \hline
	\end{tabular}
	\caption{Knowledge Bounds in a Ring; summary of results}
	\label{tbl:summary}
\end{table}

\footnotetext{\label{fn:general_graph} These results hold in general graphs, as well.}

Consider an agent $a$ at the start of a protocol given $(\alpha,\beta)$-Knowledge. 
If $a$ pretends to be a group of $d$ agents, it can be caught when $d + n - 1 > \beta$,
since agents might discover the number of agents and catch the cheater.
Moreover, \emph{any} duplication now involves some risk
since the actual value of $n$ is not known to the cheater.

An arbitrary cheating agent $a$ simulates executions of the algorithm for every possible duplication,
and evaluates its expected utility. 
Denote $D$ a duplication scheme in which an agent pretends to be $d$ agents.
Let $P_D = P[d + n - 1 \leq \beta]$ be the probability,
from agent $a$'s perspective, 
that the overall size does not exceed $\beta$. 
If for agent $a$ there exists a duplication scheme $D$ at round $0$ such that
$ \mathbb{E}_{D,0}[u_a] \cdot P_D  > \mathbb{E}_{s(0),0}[u_a]$,
then agent $a$ has an incentive to deviate and duplicate itself.

The proofs of the following theorems and corollaries
can be found in Appendix~\ref{appendix_bounds}.
For each problem we look for the maximal range of $\alpha,\beta$
where no $d$ exists that satisfies the inequation above.

\paragraph{Knowledge Sharing}
\label{class_KS}	

\begin{theorem}
	\label{theorem:ks-bound}
	Knowledge Sharing is $(2\alpha-2)$-bound.
\end{theorem}

\begin{corollary}
	\label{cor:k=2}
	$2$-Knowledge Sharing is $\infty$-bound.
\end{corollary} 

\paragraph{Coloring}
\label{class_color}

\begin{theorem}
	\label{theorem:coloring:bound}
	Coloring in a ring is $\infty$-bound.
\end{theorem}

\paragraph{Leader Election}
\label{class_LE}

In the Leader Election problem, each agent $a$ outputs $o_a \in \{0,1\}$,
where $o_a = 1$ means that $a$ was elected leader and $o_a = 0$ means otherwise.
$\Theta_L = \{O | \exists a: o_a = 1, \forall{b \neq a}: o_b = 0\}$.
We assume that every agent prefers either $0$ or $1$.

\begin{theorem}
	\label{theorem:leader-bound}
	Leader Election is $(\alpha + 1)$-bound.
\end{theorem}

\paragraph{Ring Partition}
\label{class_partition}

In the Ring Partition problem, the goal is to partition the agents of an even-sized ring
into two, equally-sized groups: group $0$ and group $1$.
We assume that every agent prefers to belong to either group $0$ or $1$.

\begin{theorem}
	\label{theorem:partition-bound}
	Ring Partition is unbounded.
\end{theorem}

\paragraph{Orientation}
\label{class_orientation}

In the Orientation problem the two ends of each edge must agree on a direction for this edge.
We assume that every agent prefers certain directions for its edges.

Unlike Ring Partition, Orientation is defined for any graph.
It is, however, a very local problem (radius $1$ in the \textit{LOCAL} model \cite{Linial:1987:DGA:1382440.1382990}).

\begin{theorem}
	\label{theorem:orientation-bound}
	The Orientation problem is unbounded.
\end{theorem}

\section{Discussion}
\label{section_discussion}

Distributed algorithms are commonly required to work in an arbitrarily large network. In a realistic scenario, the exact size of the network may not be known to all of its members. 
In this paper, we have shown that in most problems the use of duplication gives an agent power to affect the outcome of the algorithm. 
The amount of duplications an agent can create is limited by the ability of other agents to detect this deviation, and the only tool for this ability is a-priori knowledge about $n$. Section~\ref{section_imp} shows that for some problems, without any such knowledge, distributed problems become impossible to solve without any agent having an incentive to deviate from the algorithm.

The $f$-bounds we have proven for common distributed problems show that the initial knowledge required for equilibrium to be possible depends on the balance between two factors: (1) The amount of duplications necessary to increase an agent's expected utility and by how much it increases, and (2) the expected utility for an agent if it follows the protocol. In order for an agent to have an incentive to duplicate itself, an undetected duplication needs to be either a lot more profitable than following the algorithm or it must involve low risk of being caught.

Our results produce several directions that may be of interest:
\begin{enumerate}
	\item Proving impossibility and $f$-bounds in general topology graphs, as for some of the problems we only discussed ring networks. 
	
	\item Proving impossibility and showing algorithms for other problems with rational agents, which result in other tight $f$-bounds.

	\item Finding a problem that is $\alpha$-bound,
	i.e., has an equilibrium only when $n$ is known exactly.
	
	\item What defines a trivial or non-trivial problem with rational agents? More specifically, finding a characteristic that separates problems that can be solved without \emph{any} knowledge about $n$ from ones in which at least some bounds must be a-priori known.
	
	\item Finding an \emph{unbounded} problem not inherently limited (as Orientation or ring Partition are),
	or finding proof that no such problem exists.
	
	\item Exploring the effects of initial knowledge about network size in an asynchronous setting.
	
	\item Similar to \cite{DISC13/ADH,Afek:2014:DCB:2611462.2611481},
		our notion of equilibrium is basically Bayes-Nash equilibrium,
		but not sequential equilibrium \cite{KrepsWilson82}.
		Sequential equilibrium removes the assumption that agents fail the algorithm if they detect another agent cheating.
		In \cite{DISC13/ADH}, the authors suggest an additional assumption on agents' utility functions
		in order to obtain sequential equilibrium;
		however, in case an agent detects a deviation that does not necessarily lead to the algorithm failure,
		it still has no incentive to cause the algorithm failure,
		and thus it is not a sequential equilibrium.
		It would be interesting to find sequential equilibria and the problems for which they are possible.
\end{enumerate}

\section{Acknowledgment}
We would like to thank Doron Mukhtar for showing us the ring partition problem and proving it is \emph{unbounded},
when we thought such problems do not exist.
We would also like to thank Michal Feldman, Amos Fiat, and Yishay Mansour for helpful discussions.

\clearpage
\bibliographystyle{abbrv}

\clearpage
\appendix

\section{Algorithms}
\label{appendix_algs}

\subsection{Knowledge Sharing in a Ring}
\label{appendix_ks}

Here we present an algorithm for Knowledge Sharing
that is an equilibrium in a ring when no cheating agent pretends to be more than $n$ agents.
Clearly, when agents a-priori know $n$ it is an equilibrium,
since a cheating agent is further constrained not to duplicate at all.
At any point in the algorithm,
whenever an agent recognizes that another agent has deviated from the protocol
it immediately outputs $\bot$ resulting in the failure of the algorithm.

We assume a global orientation around the ring.
This assumption can be easily relaxed via Leader Election \cite{Afek:2014:DCB:2611462.2611481}.
Since the orientation has no effect on the output,
Leader Election is an equilibrium in this application.

According to Corollary~\ref{cor:ks-limit}, when a cheating agent pretends to be
\emph{more} than $n$ agents, there is no algorithm for Knowledge Sharing that is an equilibrium.
On the other hand, the algorithm presented here \emph{is} an equilibrium
when the cheating agent pretends to be \emph{no more} than $n$ agents,
proving that the bound is tight.

We start by describing the intuition behind the algorithm.
Let $n'$ be the size of the ring, which may include duplications.
Since a cheater may be at most $\lceil \frac{n'}{2} \rceil$ agents, our algorithm must ensure that any group of $\lceil \frac{n'}{2} \rceil$ consecutive agents never gains enough collective information to calculate $q(I)$, the output of Knowledge Sharing,
before the collective information at the rest of the ring is also enough to calculate $q(I)$. 

To ensure this property, we employ a method by which agent $s$ delivers its input $i_s$ to some agent $t$ at a specific round $r$,
without revealing any information about $i_s$ to any of the agents
other than $s$ and $t$. 
This method is used by every agent to send its input to two other agents
that are distant enough to prevent a consecutive group of size $\lceil \frac{n'}{2} \rceil$ from learning this input too early.

At round $n'$, the input values sent by this method are revealed simultaneously. Afterwards, every possible group of $\lceil \frac{n'}{2} \rceil$ consecutive agents had already committed to the inputs of all its members, so it is too late to change them. Now every agent can simply send its input around the ring.

\begin{figure}[H]
	\centering
	\begin{minipage}{0.38\textwidth}
		\centering
		\includegraphics[width=0.9\textwidth]{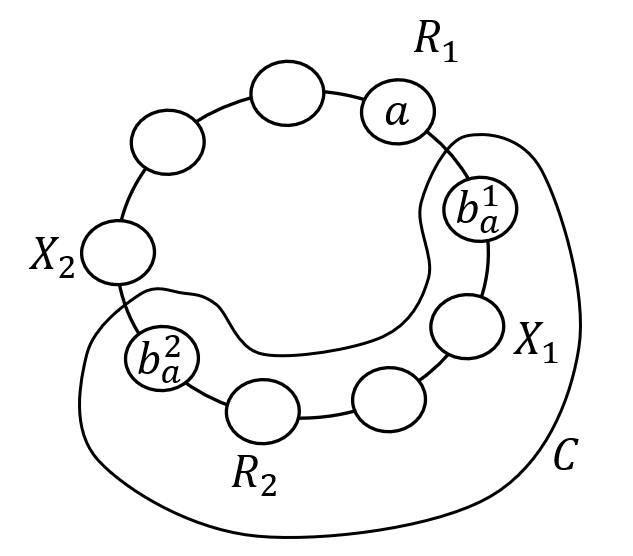} 
		\caption{at round $n'-1$: $i_a=R_1 \oplus X_1=R_2 \oplus X_2$;
			group $C$ does not know the input of $a$}
	\end{minipage}\hfill
	\begin{minipage}{0.45\textwidth}
		\centering
		\includegraphics[width=0.9\textwidth]{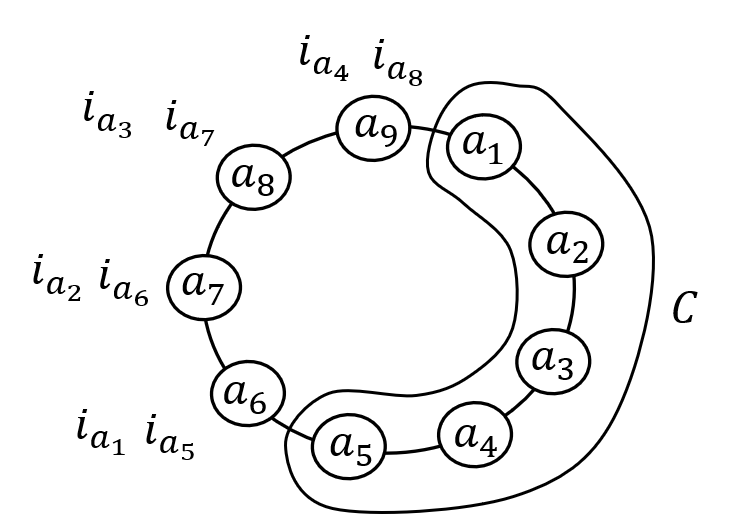}
		\caption{the input that agents $a_6,\dots,a_9$ learn at round $n'$;
			each input of an agent in $C$ is known by an agent not in $C$}
	\end{minipage}
\end{figure}

Algorithm~\ref{alg:appendix:ks-st} describes the \texttt{Secret-Transmit} building block.
The building block is called by an agent $a$ and receives three arguments:
its input $i_a$, a round $r$, and a target agent $b$.
It assumes neighbors of $b$ know they are neighbors of $b$.
The building block delivers $i_a$ at round $r$ to agent $b$,
and no other agent around the ring gains any information about this input.
Additionally, agent $b$ learns the input at round $r$ and not before.
In the \texttt{Secret-Transmit} building block, agent $a$ selects a random number $R$ and a value $X$, the XOR of its input $i_a$ with $R$.
Each value is sent in a different direction around the ring until reaching a neighbor of $b$.
At round $r-1$, both neighbors send the values $X$ and $R$ to $b$, thus $b$ learns $i_a$ at round $r$
and no other agent around the ring has any information about $i_a$ at round $r$.

\begin{algorithm}
\caption{Secret-Transmit($i_a,r,id_b$)}
\label{alg:appendix:ks-st}
\begin{algorithmic}[1]
	\item Select a random number $R$
	\item Let $X=R \oplus i_a$
	\State Send $R$ clockwise until it reaches a neighbor of $b$
	\Comment{Each message counts down from $r$}
	\Statex Send $X$ counter-clockwise until it reaches a neighbor of $b$
	\Comment{in order to know how long to wait}
	\item At round $r-1$, each neighbor of $b$ sends $b$ the value it received,
		either $X$ or $R$
\end{algorithmic}
\end{algorithm}

Algorithm~\ref{alg:appendix:ks-ring} solves Knowledge Sharing in a ring using the
\texttt{Secret-Transmit} building block.
All agents simultaneously transmit their input, each to $2$ other agents.
For each agent $a$, the input $i_a$ is sent using \texttt{Secret-Transmit} to its clockwise neighbor,
and to the agent that is at distance $\lfloor \frac{n'}{2} \rfloor$ counter-clockwise from $a$.
Note that these agents form the two ends of a group $C$ of  $\lceil \frac{n'}{2} \rceil$ consecutive agents that do not include $a$.
This guarantees that if $C$ is a cheater pretending to be $\lceil \frac{n'}{2} \rceil \leq n$ agents,
it does not learn the input $i_a$ before round $r$, since at least one piece of each transmission has not reached any agent in $C$ at any round $< r$.
At round $r$, the agents in $C$ already committed all of their input values
to some agents in the ring that are not in $C$.

\begin{algorithm}
\caption{Knowledge Sharing in a Ring}
\label{alg:appendix:ks-ring}
\begin{algorithmic}[1]
	\item All agents execute Wake-Up \cite{Afek:2014:DCB:2611462.2611481} to learn the ids
		of all agents and $n'$, the size of the ring (which may include duplications)
	\State For each agent $a$, denote $b^1_a$ the clockwise neighbor of $a$,
		and $b^2_a$ the agent at distance $\lfloor \frac{n'}{2} \rfloor$ counter-clockwise from $a$
	\State Each agent $a$ simultaneously performs:
	\Statex \texttt{SecretTransmit}($i_a,n',b^1_a$)
	\Statex \texttt{SecretTransmit}($i_a,n',b^2_a$)
	\State At round $n'+1$, each agent sends its input around the ring
	\State At round $2n'$ output $q(I)$
\end{algorithmic}
\end{algorithm}

\begin{theorem}
	\label{theorem:appendix:ks-limit}
	In a ring, Algorithm~\ref{alg:appendix:ks-ring} is an equilibrium
	when no cheating agent pretends to be more than $n$ agents.
\end{theorem}
\begin{proof}
	Assume by contradiction that a cheating agent pretending to be $d \leq n$ agents
	has an incentive to deviate from Algorithm~\ref{alg:appendix:ks-ring},
	w.l.o.g., the duplicated agents are $a_1,\dots,a_d$
	(recall the indices $1,\dots,n'$ are not known to the agents).

	Let $n'$ be the size of the ring including the duplicated agents,
	i.e., $n' = n+d-1$.
	The clockwise neighbor of $a_{n'}$ is $a_1=b^1_{a_{n'}}$.
	Denote $a_c=b^2_{a_{n'}}$ the agent at distance $\lfloor \frac{n'}{2} \rfloor$ counter-clockwise from $a_{n'}$,
	and note that $c \geq d$.
	
	When $a_{n'}$ calls \texttt{Secret-Transmit} to $a_1$,
	$a_{n'}$ holds the piece $R$ of that transmission until round $n'-1$.
	When $a_{n'}$ calls \texttt{Secret-Transmit} to $a_c$,
	$a_{c+1}$ holds the piece $X$ of that transmission until round $n'-1$.
	By our assumption, the cheating agent duplicated into $a_1,\dots,a_d$.
	Since $d < c+1$,
	the cheater receives at most one piece ($X$ or $R$)
	of each of $a_{n'}$'s transmissions before round $n'$.
	So, there is at least one input that the cheater does not learn before
	round $n'$.
	According to the Full Knowledge property (Definition~\ref{full_knowledge}),
	for the cheater at round $n'-1$ any output is equally possible,
	so its expected utility for any value it sends is the same,
	thus it has no incentive to cheat regarding the values it sends in round $n'-1$.
	
	Let $a_j \in \{a_1,\dots,a_d\}$ be an arbitrary duplicated agent.
	In round $n'$, $i_{a_j}$ is known by its clockwise neighbor $b^1_{a_j}$
	and by $b^2_{a_j}$, the agent at distance $\lfloor \frac{n'}{2} \rfloor$
	counter-clockwise from $a_j$.
	Since the number of counter-clockwise consecutive agents in
	$\{b^1_{a_j},a_j,\dots,b^2_{a_j}\}$	is greater than
	$\lceil \frac{n'}{2} \rceil \geq n$,
	at least one of $b^1_{a_j}, b^2_{a_j}$ is not a duplicated agent.
	Thus, at round $n'$, the input of each agent in $\{a_1,\dots,a_d\}$
	is already known by at least one agent $\notin \{a_1,\dots,a_d\}$.

	At round $n'-1$ the cheater does not know the input value of at least one other agent,
	so it has no incentive to deviate.
	At round $n'$ for each duplicated agent the cheating agent pretends to be,
	its input is already known by a non-duplicated agent,
	which disables the cheater from lying about its input
	from round $n'$ and on.
		
	Thus, the cheating agent has no incentive to deviate, contradicting our assumption.
\end{proof}

\subsection{Coloring Algorithm}
\label{appendix_coloring}

Here, agents are given exact a-priori knowledge of $n$, i.e., they know the exact value of $n$ at the beginning of the protocol.
We present two protocols for Coloring with rational agents, and discuss their properties.
In both algorithms, whenever an agent recognizes that another agent has deviated from the protocol,
it immediately outputs $\bot$ resulting in the failure of the algorithm.

\subsubsection{Tie Breaking}
\label{sub_tb}
In most algorithms with rational agents,
a prominent strategy \cite{Afek:2014:DCB:2611462.2611481,Halpern:2016:RCE:2933057.2933088,CRYPTO/DodisHR00} is to create a neutral mechanism that when two agents' preferences conflict,
the mechanism decides which agent gets its preference, and which does not.
We refer to such a mechanism as  \emph{tie breaking}.

Since agent $id$s are private and agents may cheat about their $id$,
they cannot be used to break ties.
However, an \emph{orientation} over an edge shared by both agents,
achieved without any agent deviating from the protocol that leads to it,
can be such a tie breaking mechanism for coloring:
whenever neighbors prefer the same color, we break ties according to the orientation
of the link between them. Breaking ties for coloring also requires the orientation to be acyclic, since a cycle in which all agents prefer the same color creates a "tie" that isn't broken by the orientation.

Note that since the agents are rational,
unless agent $a$ knows that one or more of its neighbors output its preferred color $p_a$,
it will output it itself regardless of the result of the algorithm, which is a \emph {deviation}.
Thus, any coloring algorithm must ensure that whenever an agent can output its preferred color,
it does, otherwise the agent has an incentive to deviate.

We create an acyclic orientation by a \emph{Renaming} algorithm
that reaches equilibrium \cite{Afek:2014:DCB:2611462.2611481}.
The algorithm gives new names $1,\dots,n$ to the agents, which is in fact an $n$-coloring of $G$;
however due to the circumstances described above
(each agent should output its preference if none of its neighbors does),
this $n$-coloring is not enough.
Instead, each agent, in order of the new names, picks its preferred color if available,
or the minimal available color otherwise, and sends its color to all of its neighbors.

\begin{algorithm}[H]
	\caption{Distributed Coloring Using Renaming (for agent $a$)}
	\label{alg:Renaming}
	\begin{algorithmic}[1]
		\State set $T := \varnothing$
		\State Run \emph{Renaming} subroutine
		\Statex Denote $N(v) \in \mathbb{N}$ the name of $v \in V$ from the Renaming subroutine
		\Statex Denote $t$ the first round after \emph{Renaming} was complete
		\For {round $i = t$ ;  $i <= t + n ; i++$}
		\If {$N(a) = t + i$ }
		\If {$p_a \notin T$}
		\State Send $p_a$ to all neighbors
		\State Set $o_a = p_a$
		\Else
		\State Send $m = \min_k{k \notin T}$
		\State Set $o_a = m$
		\EndIf
		\Else 
		\If {received $N(v)$ from neighbor $v$}
		\State $T = T \cup \{N(v)\}$
		\EndIf
		\EndIf
		\EndFor 
	\end{algorithmic}
\end{algorithm}

\begin{theorem}
	Algorithm~\ref{alg:Renaming} reaches Distributed Equilibrium for the coloring problem.
\end{theorem}

\begin{proof}
	Let $a$ be an arbitrary agent. Assume in contradiction that at some round $r$ there is a possible step $s \neq s_r$ such that:
	$$ \mathbb{E}_{s,r}[u_a] > \mathbb{E}_{s_r,r}[u_a]$$
	
	First, it must be shown that an agent does not have an incentive to deviate in the subroutine in order to affect the output of the entire algorithm. In the case of Algorithm~\ref{alg:Renaming}, the only deviation that would benefit $a$ in the Renaming subroutine is to minimize $N(a)$, i.e. ensuring it picks a color as early as possible. From the building block in \cite{Afek:2014:DCB:2611462.2611481} we get that the Renaming building block is an equilibrium for agents with preferences on the resulting names, thus ensuring that there is no relevant deviation possible in the subroutine as no agent can unilaterally improve the probability of having a lower $N$ value.
	
	Another property of the Renaming subroutine is that, after its completion, all agents know the names assigned to all agents in the network.
	
	Consider the possible steps $a$ could take at any round $r$ following the Renaming subroutine:
	
	\begin{itemize}
		\item Sending a message out of order is immediately recognized, as all $N$ values are known to all the agents as well as the round number. i.e., in any round $ r \neq t + N(a)$, $a$ has no incentive send any message at all, since it fails the algorithm.
		\item At $ r = t + N(a)$ , $a$ must output a color and send it to its neighbors. If $p_a \notin T$ then $a$ outputs $p_a$. It also has no incentive \emph{not} to correctly notify its neighbors that it is its output, as this notification ensures none of them output $p_a$ (as that would result in $0$ utility for that neighbor). If $p_a \in T$ then the color is taken by a neighbor, and $a$ has no incentive to deviate since its utility is already $0$.  
	\end{itemize}

	Thus, the algorithm solves Coloring and is an equilibrium.
\end{proof}

\subsubsection{Improving The Algorithm}
The \emph{Renaming} process induces more than an acyclic orientation of graph $G$,
it is a total ordering of all agents in the graph.
Coloring, however, is in many cases a local property
and can be decided locally \cite{Nati:Locality,Cole:1986:DCT:10366.10368,Goldberg:1987:PSS:28395.28429}. 
Additionally, the \emph{Renaming} protocol in \cite{Afek:2014:DCB:2611462.2611481} uses a costly $O(|E| \cdot n^2)$
message complexity.

We present another algorithm for coloring, detailed in Algorithm~\ref{alg:SemiTotalOrder},
which improves the message complexity to $O(|E|n)$
by computing an acyclic orientation of graph $G$.

First, run Wake-Up \cite{Afek:2014:DCB:2611462.2611481} to learn the graph topology
and the $id$s of all agents.
Then, in order of $id$s, each agent $a$ draws a random number $S(a)$
with a neighboring "witness" agent $w(a)$ as specified in Algorithm~\ref{alg:Draw}, and sends it to all of its neighbors.
The number is drawn in the range $1,\dots,n$ and is different than the numbers of all neighbors of $a$,
which is in fact a coloring of $G$. However, due to the circumstances described in
\ref{sub_tb}, this coloring is not enough.
By picking a random number with a witness,
the agent cannot cheat in the random number generation process, 
and $w(a)$ is marked as a witness for future verification.
When done, each agent simultaneously verifies the numbers published by its neighbors using Algorithm~\ref{alg:prompt},
which enables it to receive each value through two disjoint paths: directly from the neighbor,
and via the shortest simple path to the neighbor's witness that does not include the neighbor.
Then each agent, in order of the values drawn randomly, 
picks its preferred color if available, or the minimal available color otherwise,
and sends its color to all of its neighbors.

The resulting message complexity of the algorithm is as follows: 
$Wake-Up$ is $O(|E| \cdot n)$.
Drawing a random number is called $n$ times and thus uses $O(|E|)$ messages in total,
to publish $S$ values to neighbors.
Verifying the value of a neighbor uses $O(diameter)$ messages and is called $|E|$ times,
for a total of $O(|E| \cdot diameter)$ messages.
Sending the output color to all neighbors uses an additional $O(|E|)$ messages.
The total number of messages is thus $O(|E| \cdot n)$.
\begin{algorithm}[H]
	\caption{Draw($T$) Subroutine (for agent $a$)}
	\label{alg:Draw}
	\begin{algorithmic}[1]
		\Statex Denote $X={1,...,n}\setminus{T}$ \Comment {X is the set of numbers not drawn by neighbors}
		\State $w(a) :=$ minimal $id$ $u \in N(a)$
		\Statex \textbf{send} $witness$ to $w(a)$ 
		\Comment {choose neighbor with minimal $id$ as witness} 	
		\State $r(a) := random\{1,...,|X|\}$ drawn by $a$
		\Statex $r(w(a)) := random\{1,...,|X|\}$ drawn by $w(a)$ 
		\Statex \textbf{send} $r(a)$ to $w(a)$
		\Statex \textbf{receive} $r(w(a))$ from $w(a)$
		\Comment{$a$ and witness jointly draw a random number}
		\State Let $q := r(a)+r(w(a))\ mod\ |X|$.
		\Statex Save $S(a) := q$'th largest number in $X$
		\Statex \textbf{send} $S(a)$ to all $u \in N(a)$ \Comment {Calculate $S(a)$ and publish to neighbors} 
	\end{algorithmic}
\end{algorithm}

\begin{algorithm}[H]
	\caption{Prompt$(u)$ Subroutine (for agent $a$)}
	\label{alg:prompt}
	\begin{algorithmic}[1]
		\Statex upon receiving prompt$(u)$ from $u \in N(a)$:
		\State $p := \text{ shortest simple path  } a \rightarrow w(a) \rightarrow u$
		\Statex \textbf{send} $S(a),u $ via $p$ 
		\Comment{If $v \neq w(a)$ is asked to relay $S(a)$, $v$ fails the algorithm}
		\Statex \textbf{send} $S(a)$ to $u$ via $e = (a,u)$
		\Comment {$u$ validates that both messages received are consistent}

	\end{algorithmic}
\end{algorithm}

\begin{algorithm}[H]
	\caption{Coloring via Acyclic Orientation (for agent $a$)}
	\label{alg:SemiTotalOrder}
	\begin{algorithmic}[1]
		
		\State Run Wake-Up \Comment{After which all agents know graph topology}
		\State set $T := \varnothing$
		\For{$i=1,...,n$}
		\If{$id_a = i$'th largest $id$ in $V$}
		\State $Draw(T)$
		\Else
		\State \textbf{wait} $|Draw|$ rounds \Comment {Draw takes a constant number of rounds}	
		\If{received $S(v)$ from $v \in N(a)$}
		\State $T = T \cup \{S(v)\}$ \Comment {Add $S(v)$ to set of taken values}
		\EndIf
		\EndIf
		\EndFor
		
		\For{$u \in N(a)$ \textbf{simultaneously}}
		\State $Prompt(u)$  \Comment{Since we must validate the value received in line $8$}
		\EndFor
		\State \textbf{wait} until all prompts are completed in the entire graph \Comment {At most $n$ rounds}
		\For{round $t = 1,...,n $}:
		\If { $S(a) = t$ }  \Comment {Wait for your turn, decreed by your $S$ value}
		\If { $\forall v \in N(a) : o_v \neq p_a$}  $o_a := p_a$
		
		\Else { $o_a := $ minimum color unused by any $v \in N(a)$}
		\EndIf 
		\State \textbf{send} $o_a$ to $N(a)$
		\EndIf
		\EndFor
		
	\end{algorithmic}
\end{algorithm}

\begin{theorem}
	Algorithm~\ref{alg:SemiTotalOrder} reaches Distributed Equilibrium for the coloring problem.
\end{theorem}

\begin{proof}
	Let $a$ be an arbitrary agent. Assume in contradiction that at some round $r$ there is a possible step $s \neq s_r$ such that:
	$$ \mathbb{E}_{s,r}[u_a] > \mathbb{E}_{s_r,r}[u_a]$$

	Consider the possible deviations for $a$ in every phase of Algorithm~\ref{alg:SemiTotalOrder}:
	\begin{itemize}
		\item Cheating in Wake-Up.
		The expected utility is independent of the order by which agents draw their random number in Algorithm~\ref{alg:Draw}, i.e., the order by which agents initiate Algorithm~\ref{alg:Draw} has no effect on the order by which they will later set their colors.
		so $a$ has no incentive to publish a false $id$ in the Wake-Up building block.
		\item Drawing a random number with a witness is an equilibrium: Both agents send a random number at the same round.
		\item Publishing a false $S$ value will be caught by a future verification process with $w(a)$ when all $S$ values are verified (step 10 of Algorithm~\ref{alg:SemiTotalOrder}).
		\item Sending a color message not in order will be immediately recognized by the neighbors, since $S$ values were verified.
		\item $a$ might output a different color than the color dictated by Algorithm~\ref{alg:SemiTotalOrder}. But if the preferred color is available,
		then outputting it is the only rational behavior.
		Otherwise, the utility for the agent is already $0$ in any case.
	\end{itemize}

	Thus, the algorithm solves Coloring and is an equilibrium.	
\end{proof}

\pagebreak

\section{How Much Knowledge Is Necessary?}
\label{appendix_bounds}

\subsection{Knowledge Sharing is $(2\alpha-2)$-bound}
Here we prove Theorem~\ref{theorem:ks-bound} and Corollary~\ref{cor:k=2}.

\label{theorem:ks-bound:proof}
\begin{proof}
	
	Assume agents have $(\alpha,\beta)$-knowledge for some $\alpha,\beta$.
	A cheating agent $a$'s goal is to choose a value $d$,
	the number of agents it pretends to be, that maximizes its expected utility.
	
	Let $k$ be the number of possible outputs of a Knowledge Sharing algorithm,
	i.e., the range of the output function $g$ is of size $k$. 
	By the Full Knowledge property (definition \ref{full_knowledge}), any output is \emph{equally} possible. 
	Therefore, without deviation the expected utility of $a$ at round $0$ is:
	$\mathbb{E}_{s(0),0}[u_a] = \frac{1}{k}$.
	
	According to Theorem~\ref{theorem:ks-limit}, Algorithm~\ref{alg:appendix:ks-ring} is an equilibrium
	for Knowledge Sharing in a ring when a cheating agent pretends to be $n$ agents or less.
	Corollary~\ref{cor:ks-limit} shows that when a cheating agent
	pretends to be \emph{more} than $n$ agents,
	no algorithm for Knowledge Sharing is an equilibrium. 
	Thus, looking at all possible values of $n$ in the range $[\alpha,\beta]$,
	$a$ wants to maximize the probability that $d > n$ and the duplication increases its utility,
	while also minimizing the probability that $d+n-1 > \beta$ and the algorithm fails.

	If $a$ cheats and pretends to be $d$ agents, then necessarily $d \geq \alpha$,
	otherwise according to Theorem~\ref{theorem:ks-limit} there is an equilibrium,
	and the duplication does not increase its utility for any value of $n$ in $[\alpha,\beta]$.
	Additionally it holds that $d \leq \lceil \frac{\beta}{2} \rceil +1$,
	since any higher value of $d$ increases $a$'s chances of being caught and failing the algorithm
	(when $d+n-1 > \beta$),
	without increasing the number of possible values of $n$ for which its utility is higher
	(when $d > n$ and $d+n-1 \leq \beta$).

	From $a$'s perspective at the beginning of the algorithm,
	the value of $n$ is uniformly distributed over $[\alpha,\beta]$,
	i.e., there are a total of $\beta-\alpha+1$ equally possible values for $n$.
	According to Corollary~\ref{cor:ks-limit} when $d > n$
	agent $a$ can increase its expected utility by deviating. 
	Let $X > \frac{1}{k}$ be the utility $a$ gains by pretending to be $d > n$ agents successfully
	(i.e., when $d+n-1 \leq \beta$ and the algorithm does not fail).
	Since $n$ is uniformly distributed over $[\alpha,\beta]$ this has a probability of $\frac {d-\alpha}{\beta - \alpha +1}  $ to occur.
	On the other hand, when pretending to be $d \leq n$ agents successfully
	the utility of $a$ does not change and is $\frac{1}{k}$,
	and this has a probability of 
	$\frac{\lceil \frac{\beta}{2} \rceil - d + 1}{\beta - \alpha + 1}$.
	In all other cases $d+n-1>\beta$ and the algorithm fails, resulting in a utility of $0$.
	Thus, the expected utility of agent $a$ at round $0$:
	
	\begin{equation}
	\label{ks:bound:utility}
	\mathbb{E}_{D,0}[u_a] =
		X \cdot \frac {d-\alpha}{\beta - \alpha +1}  +
		\frac{1}{k} \cdot \frac{\lceil \frac{\beta}{2} \rceil - d + 1}{\beta - \alpha + 1}
	\end{equation}
	
	By the constraints specified above, the value of $d$ that maximizes (\ref{ks:bound:utility})
	is $d = \lfloor \frac{\beta}{2} \rfloor + 1$,
	and $a$ will deviate from the algorithm whenever $\mathbb{E}_{D,0}[u_a] > \frac{1}{k}$.
	To find the $f$-bound on $[\alpha,\beta]$ we derive $\beta$ as a function of $\alpha$:

	$$ \mathbb{E}_{D,0}[u_a] = 
		X \cdot \frac {\lfloor \frac{\beta}{2} \rfloor + 1-\alpha}{\beta - \alpha +1} +
		\frac{1}{k} \cdot
		\frac{\lceil \frac{\beta}{2} \rceil - \lfloor \frac{\beta}{2} \rfloor }{\beta - \alpha + 1} >
		\frac{1}{k}$$
	
	\begin{equation}
	\label{eq_ks_expected_utility}
		\begin{cases}
			\beta \text{ is even} & \beta(kX - 2) > 2\alpha kX - 2kX -2\alpha +2 \\
			\beta \text{ is odd} &  \beta(kX - 2) > 2\alpha kX - kX -2\alpha	
		\end{cases}
	\end{equation}
	
	From (\ref{eq_ks_expected_utility}) we can derive the following:
	\begin{enumerate}
		\item The inequations are satisfiable only if $X > 2 \cdot \frac{1}{k}$.
		Since $X \leq 1$, $2$-Knowledge Sharing ($k=2$) cannot satisfy the inequations
		and $a$ never has an incentive to deviate, given \emph{any} bound.
		This proves Corollary~\ref{cor:k=2}.
		
		\item For Knowledge Sharing, we find the range $[\alpha, \beta]$ that holds for any $k$.
		As $k$ grows large, $\frac{1}{k}$ nears $0$. 
		Assuming the profit when duplication is successful is $X=1$,
		agent $a$ has an incentive to deviate when
		$ \lfloor \frac{\beta}{2} \rfloor + 1 - \alpha > 0 $.
		When $\beta$ is even: $\beta > \alpha - 2$ , and when $\beta$ is odd: $\beta > \alpha - 1$.
		Thus, Algorithm~\ref{alg:appendix:ks-ring} is an equilibrium for Knowledge Sharing when agents have
		$(\alpha,\beta)$-knowledge such that $\beta \leq \alpha-2$,
		and for any algorithm for Knowledge Sharing there exist $\alpha,\beta>\alpha-2$ such
		that there is no equilibrium when agents have $(\alpha,\beta)$-knowledge.
		This proves Theorem~\ref{theorem:ks-bound}.
	\end{enumerate}
\end{proof}

\subsection{Coloring is $\infty$-bound}
\label{theorem:coloring:proof}

Here we prove Theorem~\ref{theorem:coloring:bound}.

\begin{proof}
	Consider Algorithm~\ref{algorithm:coloring:ring} which solves coloring in a ring
	using $2$-Knowledge Sharing.
	
	\begin{algorithm}
	\caption{Coloring in a Ring}
	\label{algorithm:coloring:ring}
	\begin{algorithmic}[1]
		\item Wake-Up to learn the size of the ring
		\item Assume an arbitrary global direction over the ring
			(this can be relaxed via Leader Election
			\cite{Afek:2014:DCB:2611462.2611481})
		\item Run $2$-Knowledge Sharing to randomize a single global bit $b \in \{0,1\}$
		\item Publish the preferred color of each agent simultaneously over the entire ring
		\item In each group of consecutive agents that prefer the same color,
			if $b=0$ the even agents (according to the orientation) output their preferred color,
			else the odd agents do.
		\item If an agent has no neighbors who prefer the same color,
			it outputs its preferred color.
		\item Any other agent outputs the minimal available color.
	\end{algorithmic}
	\end{algorithm}
	
	It is easy to see that Algorithm~\ref{algorithm:coloring:ring} is an equilibrium
	and results in a legal coloring of the ring.
	It uses $2$-Knowledge Sharing and thus, following Corollary~\ref{cor:k=2},
	it proves Theorem~\ref{theorem:coloring:bound}.
\end{proof}

\subsection{Leader Election is $(\alpha+1)$-bound}
\label{theorem:le:proof}

Here we prove Theorem~\ref{theorem:leader-bound}.

\begin{proof}
	Recall that any Leader Election algorithm must be \emph{fair} \cite{DISC13/ADH}, i.e.,
	every agent must have equal probability of being elected leader for the algorithm to be an equilibrium. 
	
	Given $f(\alpha) = \alpha + 1$, the actual number of agents $n$
	is either $\alpha$ or $\alpha + 1$, decided by some distribution unknown to the agents.
	If an agent follows the protocol, the probability of being elected is $\frac{1}{n}$.
	If it duplicates itself once, the probability that a duplicate is elected is $\frac{2}{n+1}$,
	but if $n = \alpha+1$ the protocol fails and the utility is $0$.
	Thus $\mathbb{E}_{s_{dup},a}[u_a] = \frac{1}{2}  \frac{2}{n+1} < \frac{1}{n}$,
	i.e., no agent has an incentive to deviate. 
	
	Given $f(\alpha) = \alpha + 2$, then $n$ is in $[\alpha,\alpha + 2]$.
	If an agent follows the protocol, its expected utility is still $\frac{1}{n}$. 
	If it duplicates itself once, the probability that a duplicate is elected is still $\frac{2}{n+1}$,
	however \emph{only} if $n = \alpha + 2$ the protocol fails. 
	Thus, $\mathbb{E}_{s_{dup},a}[u_a] = \frac{2}{3} \frac{2}{n+1} > \frac{1}{n}$ for any $n > 3$. 
	So the agent \emph{has} an incentive to deviate.
	
	Thus for $f(\alpha) = \alpha + 1$ the algorithm presented in \cite{DISC13/ADH} is an equilibrium, 
	while for $f(\alpha) = \alpha + 2$ no algorithm for Leader Election is an equilibrium, 
	since any algorithm must be \emph{fair}.
\end{proof}

\subsection{Ring Partition is Unbounded}
\label{theorem:partition:proof}

Here we prove Theorem~\ref{theorem:partition-bound}.

\begin{proof}
	
	It is clear that an agent will not duplicate itself to change the parity of the graph,
	as that will necessarily cause an erroneous output.
	So it is enough to show an algorithm that is an equilibrium
	for even graphs, when agents have no knowledge about $n$.
	Consider the following algorithm:
	\begin{itemize}
		\item Either one arbitrary agent wakes up or we run a Wake-Up subroutine and then
		Leader Election \cite{Afek:2014:DCB:2611462.2611481}.
		Since the initiator (leader) has no effect on the output,
		both are an equilibrium in this application.
		\item The initiating agent sends a token which alternatively marks agents by 0 or 1 and also defines the direction of communication in the ring.
		\item Upon reception of the token with value $m$, an agent $a$ does one of the following:
		\begin{enumerate}
			\item If $m = 0$, send predecessor (denoted $p(a)$) a random bit $t_a$.
			\item Else, if $m = 1$, wait for 1 round and send successor (denoted $s(a)$) a random bit $t_a$.
		\end{enumerate}
		\item Upon reception of the neighbor's bit (one round after receiving the token), set 
		$$ o_a = (t_a + t_{s(a)/p(a)} + m)_{mod2} $$
		\item As the token arrives back at the initiator, it checks the token's parity. For even rings, it must be the opposite value from the value it originally sent.
	\end{itemize}
	
	This algorithm divides every pair of agents to one with output $1$ and one with output $0$, as the token value $m$ is different, thus achieving a partition.
	
	We show that it is also an equilibrium.	
	Assume an agent $a$ deviates at some round $r$. If $r$ is in the Wake-Up or Leader Election
	phase in order to be the initiator, it cannot improve its utility since choosing the starting value of the token, choosing the direction, or being first cannot increase the agent's utility. 
	If it is a deviation while the token traverses other parts of the graph, any message $a$ sends will eventually be discovered, as the real token has either already passed or will eventually pass through the "cheated" agent.
	If $a$ changes the value of the token, a randomization between two agents will be out of sync eventually at the end of the token traversal, and also the initiator will recognize that the ring does not contain an even number of agents.
	During the exchange of $t_a$ the result is independent of $a$'s choice of value for $t_a$.
	So there is no round in which $a$ can deviate from the protocol.
\end{proof}

\subsection{Orientation is Unbounded}
\label{theorem:orientation:proof}

Here we prove Theorem~\ref{theorem:orientation-bound}.

\begin{proof}
	We show a simple algorithm and prove that it is an equilibrium without any a-priori knowledge of $n$ or bounds on $n$.
	Assuming all agents start at the same round (otherwise run Wake-Up), consider the following algorithm:
	\begin{itemize}
		\item Each agent simultaneously send a random number $(0,1)$ and its $id$ on each of its edges.
		\item For each edge, XOR the bit you sent and the bit received over that edge
		\item If the result is 1, the edge is directed towards the agent with the higher $id$, otherwise it is directed towards the lower $id$.
		\item Every agent outputs the list of pairs with $id$ and direction for each of its neighbors. 
	\end{itemize}
	
	Since an agent's utility is defined over its personal output, Solution Preference inhibits agents to output a correct set of pairs, so a cheater may only influence the direction of the edges.
	Since duplication does not create any new edges between the cheater and the original graph, and the orientation is decided over each edge independently, it does not effect any agent's utility. Other than that, randomizing a single bit over an edge at the same round is in equilibrium. So the algorithm is an equilibrium, and Orientation is unbounded.
\end{proof}

\end{document}